\documentclass[letterpaper, 10 pt, conference]{ieeeconf}
\IEEEoverridecommandlockouts 
\usepackage{mypreamble}

\title{\bf \LARGE
Scaled Relative Graph Analysis of Lur'e Systems and the Generalized Circle Criterion
}

\author{Julius P.J. Krebbekx$^1$, Roland Tóth$^{1,2}$, Amritam Das$^1$
\thanks{$^1$Control Systems group, Department of Electrical Engineering,  Eindhoven University of Technology, The Netherlands.}
\thanks{$^2$Systems and Control Lab, HUN-REN Institute for Computer Science and Control, Budapest, Hungary. 
E-mail: {\tt\small \{j.p.j.krebbekx, R.Toth, am.das\}@tue.nl}}
}
\date{August 2024}

\begin{document}

\maketitle

\begin{abstract}
    Scaled Relative Graphs (SRGs) provide a novel graphical frequency-domain method for the analysis of nonlinear systems. However, we show that the current SRG analysis suffers from a pitfall that limit its applicability in analyzing practical nonlinear systems. We overcome this pitfall by modifying the SRG of a linear time invariant operator, combining the SRG with the Nyquist criterion, and apply our result to Lur'e systems. We thereby obtain a generalization of the celebrated circle criterion, which deals with a broader class of nonlinearities, and provides (incremental) $L_2$-gain performance bounds.
\end{abstract}


\section{Introduction}

In the case of a Linear Time Invariant (LTI) system, graphical system analysis using the Nyquist diagram~\cite{nyquistRegenerationTheory1932} is the cornerstone of control engineering. It is easy to use, and allows for intuitive analysis and controller design methods. However, it is currently unclear how to generalize such graphical frequency domain methods to nonlinear system analysis and controller design with their full potential. 

There have been various attempts in extending graphical analysis methods for nonlinear systems, but they are all either approximate, or limited in range of applicability. Classical results such as the circle criterion \cite{sandbergFrequencyDomainConditionStability1964}, which is based on the Nyquist concept, and the Popov criterion can predict the stability of a class of nonlinear systems. These methods are exact, but are not useful for performance shaping. Moreover, they are limited to Lur'e systems with sector bounded nonlinearities. The describing function method \cite{krylovIntroductionNonlinearMechanics1947} is an approximate method based on the Nyquist diagram, which considers only the first Fourier coefficient of the amplitude-dependent frequency response. A more sophisticated approximate method is the nonlinear bode diagram in~\cite{pavlovFrequencyDomainPerformance2007}. The Scaled Relative Graph (SRG) \cite{ryuScaledRelativeGraphs2022} proposed in~\cite{chaffeyGraphicalNonlinearSystem2023} is a new graphical method to analyze nonlinear feedback systems. It is an exact method and it is intuitive because of its close connection to the Nyquist diagram. Moreover, SRG analysis can provide performance bounds in terms of (incremental) $L_2$-gain. This has been demonstrated recently in \cite{vandeneijndenScaledGraphsReset2024}, where the framework of SRGs has been applied to the analysis of reset controllers. By leveraging the restriction of the SRG to specific input spaces, it is also possible to compute a nonlinear Bode diagram and define a bandwidth for nonlinear systems~\cite{krebbekxNonlinearBandwidthBode2025}.

Even though existing SRG tools are exact, they are limited in range of applicability, as they can only deal with stable open-loop plants. In practice, however, it is often required to stabilize an unstable plant. In this paper, we demonstrate a fundamental pitfall of SRG analysis when applying it to unstable LTI systems in a feedback interconnection. We resolve this pitfall by including the information provided by the Nyquist criterion into the SRG of the LTI operator to obtain an effective tool for computing stability conditions. We then apply our solution to the Lur'e system, for which we obtain performance metrics in terms of (incremental) $L_2$-gain bounds, and well-posedness is guaranteed via the homotopy construction~\cite{megretskiSystemAnalysisIntegral1997,chaffeyHomotopyTheoremIncremental2024}. An important consequence of our results is a generalized circle criterion, where the class of nonlinear operators is now general, instead of being limited to sector bounded nonlinearities as in the current sate-of-the-art results.

This paper is structured as follows. In Section~\ref{sec:preliminaries}, we present the required preliminaries, and in Section~\ref{sec:nl_system_analysis} the SRG concept of nonlinear system analysis. An important pitfall of the current SRG methods is identified in Section~\ref{sec:pitfall}, and its resolution is given in Section~\ref{sec:resolution}, which implies a rather general extension of the celebrated circle criterion. We apply our main result to a plant with an unstable pole in feedback with a nonlinearity in Section~\ref{sec:examples} and present our conclusions in Section~\ref{sec:conclusion}.





\section{Preliminaries}\label{sec:preliminaries}


\subsection{Notation and Conventions}

Let $\R, \C$ denote the real and complex number fields, respectively, with $\R_{>0} = (0, \infty)$ and $\C_{\mathrm{Re} > 0}= \{ a+ jb \mid \, a \in \R_{>0}, \, b \in \R \}$, where $j$ is the imaginary unit. We denote the complex conjugate of $z = a + jb \in \C$ as $\bar{z} = a-jb$. Let $\mathcal{L}$ denote a Hilbert space, equipped with an inner product $\inner{\cdot}{\cdot}_\mathcal{L} : \mathcal{L} \times \mathcal{L} \to \C$ and norm $\norm{x}_\mathcal{L} := \sqrt{\inner{x}{x}_\mathcal{L}}$.  For sets $A, B \subset \C$, the sum and product sets are defined as $A+B:= \{ a+b \mid a\in A, b\in B\}$ and $AB:= \{ ab \mid a\in A, b\in B\}$, respectively. The disk in the complex plane is denoted $D_r(x) = \{ z \in \C \mid |z-x| \leq r \}$. Denote $D_{[\alpha, \beta]}$ the disk in $\C$ centered on $\R$ which intersects $\R$ in $[\alpha, \beta]$. The radius of a set $\mathcal{C} \subset \C$ is defined by $\rmin(\mathcal{C}) := \inf_{r>0} : \mathcal{C} \subset D_r(0)$. The distance between two sets $\mathcal{C}_1,\mathcal{C}_2 \subset \C \cup \{ \infty \}$ is defined by $\dist(\mathcal{C}_1,\mathcal{C}_2) := \inf_{z_1 \in \mathcal{C}_1, z_2 \in \mathcal{C}_2} |z_1-z_2|$, where $|\infty-\infty|:=0$.

\subsection{Signals, Systems and Stability}



Since this work focuses on SISO continuous-time systems, the Hilbert space of particular interest is $L_2:= \{ f:\R_{\geq 0} \to \R \mid \norm{f}_2 < \infty \}$, where the norm is induced by the inner product $\inner{f}{g}:= \int_\R f(t) g(t) d t$. For any $T \in \R_{\geq 0}$, define the truncation operator $P_T$ as 
\begin{equation}
    (P_T u)(t) :=
    \begin{cases}
        u(t) & t \leq T, \\        
        0 & t > T.
    \end{cases}
\end{equation}
The extension of $L_2$, see Ref.~\cite{desoerFeedbackSystemsInputoutput1975}, is defined as 
\begin{equation*}
    \Lte := \{ u : \R_{\geq 0} \to \R \mid \norm{P_T u}_2 < \infty \text{ for all } T \in \R_{\geq 0} \}.
\end{equation*}
Note that the extension is particularly useful since it includes periodic signals, which are otherwise excluded from $L_2$. 

Systems are modeled as operators $R: \Lte \to \Lte$. A system is said to be causal if it satisfies $P_T (Ru) = P_T(R(P_Tu))$, i.e., the output at time $t$ is independent of the signal at times greater than $t$. Unless specified otherwise, \emph{we will always assume causality.}

Given an operator $R$ on $L_2$, the induced incremental norm of the operator is defined (similar to the notation in~\cite{vanderschaftL2GainPassivityTechniques2017}) as 
\begin{equation}\label{eq:incremental_induced_norm}
    \Gamma(R) := \sup_{u_1, u_2 \in L_2} \frac{\norm{Ru_1-Ru_2}_2}{\norm{u_1-u_2}_2}.
\end{equation}
Similarly, we define the induced non-incremental norm of the operator as
\begin{equation}\label{eq:non_incremental_induced_norm}
    \gamma(R) := \sup_{u \in L_2} \frac{\norm{Ru}_2}{\norm{u}_2}.
\end{equation}

For causal systems, the induced (non-)incremental operator norm on $L_2$ carries over to  $\Lte$ since $\norm{P_T(Ru)}_2 = \norm{P_T(R(P_Tu))}_2 \leq \norm{R(P_Tu)}_2$ and $P_T u \in L_2$ for all $u \in \Lte$ (see Ref.~\cite{vanderschaftL2GainPassivityTechniques2017}). We define the (non-)incremental $L_2$-gain of a causal operator $R : \Lte \to \Lte$ as $\Gamma(R)$ ($\gamma(R)$), i.e., the induced (non-)incremental operator norm from Eq.~\eqref{eq:incremental_induced_norm} (Eq.~\eqref{eq:non_incremental_induced_norm}). A causal system $R$ is said to be (non-)incrementally $L_2$-stable if $\Gamma(R) < \infty$ ($\gamma(R) < \infty$). A system $R$ is called $L_2$-stable if $\norm{u}_2 < \infty$ implies $\norm{Ru}_2 < \infty$. If $R(0)=0$ holds, then incremental $L_2$-stability implies $L_2$-stability.

\subsection{Scaled Relative Graphs}\label{sec:srg_definitions}

We now turn to the definition and properties of the Scaled Relative Graph (SRG), as introduced by Ryu et al. in~\cite{ryuScaledRelativeGraphs2022}. We follow closely the exposition of the SRG as given by Chaffey et al. in~\cite{chaffeyGraphicalNonlinearSystem2023}. 

\subsubsection{Definitions}

Let $\mathcal{L}$ be a Hilbert space, and $R : \mathcal{L} \to \mathcal{L}$ an operator. The angle between $u, y\in \mathcal{L}$ is defined as 
\begin{equation}\label{eq:def_srg_angle}
    \angle(u, y) := \cos^{-1} \frac{\mathrm{Re}\inner{u}{y}}{\norm{u} \norm{y}} \in [0, \pi].
\end{equation}
Given $u_1, u_2 \in \mathcal{U} \subset \mathcal{L}$, we define the set of complex numbers
\begin{multline*}
    z_R(u_1, u_2) := \left\{ \frac{\norm{Ru_1-Ru_2}}{\norm{u_1-u_2}} e^{\pm j \angle(u_1-u_2, Ru_1-Ru_2)} \right\}.
\end{multline*}
The SRG of $R$ over the set $\mathcal{U}$ is defined as
\begin{equation*}
    \SRG_\mathcal{U} (R) := \bigcup_{u_1, u_2 \in \mathcal{U}} z_R(u_1, u_2).
\end{equation*}
When $\mathcal{U}=\mathcal{L}$, we denote $\SRG_\mathcal{L}(R) = \SRG(R)$. Note that the SRG is a subset of $\C$.

One can also define the Scaled Graph (SG) around some particular input. The SG of an operator $R$ with one input $u^\star \in \mathcal{L}$ fixed and the other in set $\mathcal{U} \subset \mathcal{L}$ is defined as
\begin{equation}
    \SG_{\mathcal{U}, u^\star}(R) := \{ z_R(u, u^\star) \mid u \in \mathcal{U} \}.
\end{equation}
We introduce the shorthand $\SG_{\mathcal{L}, u^\star}(R) = \SG_{u^\star}(R)$.


By definition of the SRG (SG), the (non-)incremental gain of a system $R$, defined in Eq.~\eqref{eq:incremental_induced_norm} (Eq.~\eqref{eq:non_incremental_induced_norm}), is equal to the radius of the SRG (SG at zero), i.e. $\Gamma(R) = \rmin(\SRG(R))$ ($\gamma(R) = \rmin(\SG_0(R))$).

\subsubsection{Operations on SRGs}\label{sec:operations_on_srgs}

The facts presented here are proven in~\cite[Chapter 4]{ryuScaledRelativeGraphs2022}. 

Inversion of a point $z = re^{j \phi} \in \C$ is defined as the M\"obius inversion $r e^{j \phi} \mapsto (1/r)e^{j \phi}$. An operator $R$ satisfies the \emph{chord property} if, for all $z \in \SRG(R) \setminus \{ \infty \}$, it holds that $[z, \bar{z}] \subset \SRG(R)$ (see~\cite[Ch. 4.4]{ryuScaledRelativeGraphs2022}). An operator $R$ is said to satisfy the left-hand (right-hand) arc property if for all $z \in \SRG(R)$, it holds that $\operatorname{Arc}^-(z, \bar{z}) \subset \SRG(R)$ ($\operatorname{Arc}^+(z, \bar{z}) \subset \SRG(R)$), where $\operatorname{Arc}^-,\operatorname{Arc}^+$ are defined in~\cite[Ch. 4.5]{ryuScaledRelativeGraphs2022}. If $R$ satisfies the left-hand, right-hand, or both arc properties, it is said to satisfy \emph{an} arc property.

\begin{proposition}\label{prop:srg_calculus}
    Let $0 \neq \alpha \in \R$ and let $R, S$ be arbitrary operators on the Hilbert space $\mathcal{L}$. Then, 
    \begin{enumerate}[label=\alph*.]
        \item\label{eq:srg_calculus_alpha} $\SRG(\alpha R) = \SRG(R \alpha) = \alpha \SRG(R)$,
        \item\label{eq:srg_calculus_plus_one} $\SRG(I + R) = 1 + \SRG(R)$, where $I$ denotes the identity on $\mathcal{L}$,
        \item\label{eq:srg_calculus_inverse} $\SRG(R^{-1}) = (\SRG(R))^{-1}$.
        \item\label{eq:srg_calculus_parallel} If at least one of $R, S$ satisfies the chord property, then $\SRG(R + S) \subset \SRG(R) + \SRG(S)$.
        \item\label{eq:srg_calculus_series} If at least one of $R, S$ satisfies an arc property, then $\SRG(R S) \subset \SRG(R) \SRG(S)$.
    \end{enumerate}
    If the SRGs above contain $\infty$ or are the empty set, the above operations are slightly different, see~\cite{ryuScaledRelativeGraphs2022}. 
\end{proposition}

\section{Nonlinear System Analysis}\label{sec:nl_system_analysis}

In this section, we review the state-of-the-art methods for analyzing nonlinear feedback systems. We start with the LTI case, where the Nyquist stability criterion is discussed, since it plays central role in the analysis of nonlinear feedback systems of Lur'e form. 

So far, $R$ represents a general system. From now on, we will use $G$ for the plant, $K$ for the controller, $L=GK$ for the loop transfer, and $T$ for the closed-loop. If a system is LTI, it is often denoted by a Laplace argument. 

\subsection{The Nyquist Criterion}

\begin{figure}[tb]
    \centering

    \tikzstyle{block} = [draw, rectangle, 
    minimum height=2em, minimum width=2em]
    \tikzstyle{sum} = [draw, circle, node distance={0.5cm and 0.5cm}]
    \tikzstyle{input} = [coordinate]
    \tikzstyle{output} = [coordinate]
    \tikzstyle{pinstyle} = [pin edge={to-,thin,black}]
    
    \begin{tikzpicture}[auto, node distance = {0.3cm and 0.5cm}]
        \node [input, name=input] {};
        \node [sum, right = of input] (sum) {$\Sigma$};
        \node [block, right = of sum] (lti) {$L(s)$};
        \node [coordinate, right = of lti] (z_intersection) {};
        \node [output, right = of z_intersection] (output) {}; 
        \node [coordinate, below = of lti] (static_nl) {};
    
        \draw [->] (input) -- node {$r$} (sum);
        \draw [->] (sum) -- node {$e$} (lti);
        \draw [->] (lti) -- node [name=z] {$y$} (output);
        \draw [-] (z) |- (static_nl);
        \draw [->] (static_nl) -| node[pos=0.99] {$-$} (sum);
    \end{tikzpicture}
    
    \caption{A simple linear feedback system.}
    \label{fig:linear_feedback}
    \vspace{-0.5em}
\end{figure}
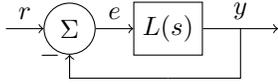

Consider the LTI feedback system in Fig.~\ref{fig:linear_feedback}, where $L(s)$ is a transfer function. In many control engineering situations, one is interested in the stability and performance aspects of this setup. In this work, we study the stability property.

\begin{theorem}\label{thm:nyquist}
    Let $n_\mathrm{z}$ denote the number of unstable closed-loop poles, and $n_\mathrm{p}$ the number of unstable poles of $L(s)$. Additionally, denote $n_\mathrm{n}$ the amount of times that $L(j \omega)$ encircles the point $-1$ in clockwise fashion as $\omega $ traverses the $D$-contour, going from $-jR$ to $jR$ and then along $Re^{j \phi}$ as $\phi$ goes from $\pi/2 \to -\pi/2$, for $R \to \infty$. The closed-loop system in Fig.~\ref{fig:linear_feedback} satisfies $n_\mathrm{z}=n_\mathrm{n}+n_\mathrm{p}$.
\end{theorem}

Note that for stability, $n_\mathrm{z}=0$ is required.

\subsection{Nonlinear Feedback Systems}\label{sec:nonlin_feedback_systems}
The Lur'e system, as depicted in Fig.~\ref{fig:lure} consists of a SISO LTI block $G(s)$ connected with a static nonlinear function $\phi : \R \to \R$. The closed-loop system cannot be represented by a transfer function anymore, but is instead an operator $T : \Lte \to \Lte$. The closed loop operator, also known as complementary sensitivity operator, can be written as
\begin{equation}\label{eq:lure_closedloop}
    T = (G^{-1} + \phi)^{-1},
\end{equation}
which is derived via $y=G(r-\phi y) \iff G^{-1} y = r - \phi y \iff (G^{-1} + \phi) y=r$, where the nonlinearity of $\phi$ is respected. Similarly, the sensitivity can be written as 
\begin{equation}
    S = (1 + \phi G)^{-1},
\end{equation}
which is derived via $e=r-\phi y \iff e= r-\phi G e \iff (1 + \phi G) e=r$, where the multiplication order $\phi G$ is important due to the nonlinearity of $\phi$.

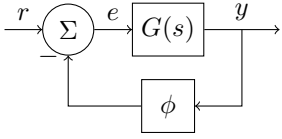
\begin{figure}[tb]
    \centering

    \tikzstyle{block} = [draw, rectangle, 
    minimum height=2em, minimum width=2em]
    \tikzstyle{sum} = [draw, circle, node distance={0.5cm and 0.5cm}]
    \tikzstyle{input} = [coordinate]
    \tikzstyle{output} = [coordinate]
    \tikzstyle{pinstyle} = [pin edge={to-,thin,black}]
    
    \begin{tikzpicture}[auto, node distance = {0.3cm and 0.5cm}]
        \node [input, name=input] {};
        \node [sum, right = of input] (sum) {$\Sigma$};
        \node [block, right = of sum] (lti) {$G(s)$};
        \node [coordinate, right = of lti] (z_intersection) {};
        \node [output, right = of z_intersection] (output) {}; 
        \node [block, below = of lti] (static_nl) {$\phi$};
    
        \draw [->] (input) -- node {$r$} (sum);
        \draw [->] (sum) -- node {$e$} (lti);
        \draw [->] (lti) -- node [name=z] {$y$} (output);
        \draw [->] (z) |- (static_nl);
        \draw [->] (static_nl) -| node[pos=0.99] {$-$} (sum);
    \end{tikzpicture}
    
    \caption{Block diagram of a Lur'e system.}
    \label{fig:lure}
    \vspace{-0.5em}
\end{figure}

We call a Lur'e system \emph{well-posed} if $r \mapsto (e, y)$ satisfies existence, uniqueness and continuity. The system is called bounded if $r \to (e, y)$ is bounded, which is equivalent to bounded-input-bounded-output (BIBO) stability. If $R$ has finite incremental $L_2$-gain, additionally $R(0)=0$ is required for BIBO stability. Note that \emph{causality} is not part of our well-posedness definition, and has to be assumed throughout.

In order to assess stability of a Lur'e system, one can use the circle criterion, see~\cite[Thm. 5.2.10]{desoerFeedbackSystemsInputoutput1975}.

\begin{theorem}\label{thm:circle}
    Let $G(s)$ be a strictly proper transfer function and let $\phi \in [k_1, k_2]$, meaning $k_1 \leq \phi(x)/x \leq k_2, \quad \forall x\in \R$. Let $n_\mathrm{p}$ be the number of poles $p$ of $G(s)$ such that $\mathrm{Re}(p)>0$. Then, the system in Fig.~\ref{fig:lure} is $L_2$-stable if it satisfies one of the conditions:
    \begin{enumerate}
        \item Let $0<k_1<k_2$. The Nyquist diagram of $G(s)$ must not intersect $D_{[-1/k_1, -1/k_2]}$ and has to encircle it $n_p$ times in counterclockwise direction.
        \item if $0=k_1<k_2$, then $n_\mathrm{p}=0$ must hold and the Nyquist diagram must satisfy 
        \begin{equation*}
            \mathrm{Re} \,G(j\omega) > -1/k_2, \quad \forall \omega \in \R.
        \end{equation*}
        \item if $k_1 <0 <k_2$, then $n_\mathrm{p}=0$ must hold and the Nyquist diagram has to be contained in the interior of $D_{[-1/k_1, -1/k_2]}$.
    \end{enumerate}
\end{theorem}

We denote by $\partial \phi \in [k_1, k_2]$, if the sector condition is satisfied incrementally.


\subsection{System Analysis with Scaled Relative Graphs}\label{sec:sys_analysis_w_srgs}

Chaffey et al.~\cite{chaffeyGraphicalNonlinearSystem2023}, has applied the SRG for the first time to the analysis of feedback systems. In this section, we introduce the SRG of various system components.

Define the Nyquist diagram as $\operatorname{Nyquist}(G) = \{ G(j \omega)\mid \omega \in \R\}$.

\begin{theorem}\label{thm:lti_srg}
    Let $R : \Lte \to \Lte$ be stable and LTI with transfer function $G_R(s)$, then $\SRG(R) \cap \C_{\mathrm{Im} \geq 0}$ is the h-convex hull of $\operatorname{Nyquist}(G_R) \cap \C_{\mathrm{Im} \geq 0}$. The h-convex hull of a set $A$ is obtained by adding all circle segments between $z_1\in A$ and $z_2 \in A$, where the circle is centered on $\R$ and passes through $z_1$ and $z_2$. 
\end{theorem}


\begin{proposition}\label{prop:static_nl_srg}
    If $\partial \phi \in [k_1, k_2]$, i.e. $\phi$ satisfies an incremental sector condition, then one has 
    \begin{equation*}
        \SRG(\phi) \subset D_{[k_1, k_2]}.
    \end{equation*}
    Furthermore, if there is a point at which the slope of $\phi$ switches in a discontinuous fashion from $k_1$ to $k_2$, then the inclusion becomes an equality.
\end{proposition} 

For proofs of Theorem~\ref{thm:lti_srg} and Proposition~\ref{prop:static_nl_srg}, see~\cite{chaffeyGraphicalNonlinearSystem2023}. In order to compare our results to the circle criterion (Theorem~\ref{thm:circle}) later in Section~\ref{sec:compare_circle_criterion}, we prove the following important extension of Proposition~\ref{prop:static_nl_srg}.

\begin{proposition}\label{prop:static_nl_srg_non_incremental}
    If $\phi \in [k_1, k_2]$, i.e., $\phi$ satisfies a sector condition, then one has 
    \begin{equation*}
        \SG_{0}(\phi) \subset D_{[k_1, k_2]}.
    \end{equation*}
\end{proposition}

\begin{proof}
    Exactly the same as the proof of~\cite[Proposition 9]{chaffeyGraphicalNonlinearSystem2023}, where one picks $u_2 = 0$.
\end{proof}


We have the following theorem from~\cite{chaffeyGraphicalNonlinearSystem2023}, updated in~\cite{chaffeyHomotopyTheoremIncremental2024}, for any system $H_1$ feedback interconnected with $H_2$, as displayed in Fig.~\ref{fig:chaffey_thm2}.

\begin{theorem}\label{thm:chaffey_thm2}
    Consider $H_1, H_2$ be operators on $\Lte$, where $\Gamma(H_1) < \infty$ and $H_2$ satisfies for all $\tau \in (0, 1]$
    \begin{equation*}
        \dist(\SRG(H_1)^{-1}, -\tau \SRG(H_2)) \geq r_m >0,
    \end{equation*}
    and at least one of $H_1, H_2$ obeys the chord property. Then, the feedback connection in Fig.~\ref{fig:chaffey_thm2} has an incremental $L_2$-gain bound of $1/r_m$ and is well-posed~\cite{chaffeyHomotopyTheoremIncremental2024}. Upon replacing $\SRG$ with $\SG_0$, one obtains a non-incremental $L_2$-gain bound. However, no conclusion can be made regarding the well-posedness of the closed loop system. 
\end{theorem}

The requirement that $H_1$ is stable poses a severe limitation for the applicability of SRG methods to controller design, as it is impossible to stabilize an unstable plant using Theorem~\ref{thm:chaffey_thm2}. Removing this limitation is the first key result of this paper.

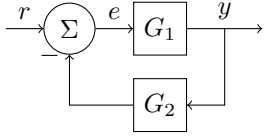
\begin{figure}[t]
    \centering

    \tikzstyle{block} = [draw, rectangle, 
    minimum height=2em, minimum width=2em]
    \tikzstyle{sum} = [draw, circle, node distance={0.5cm and 0.5cm}]
    \tikzstyle{input} = [coordinate]
    \tikzstyle{output} = [coordinate]
    \tikzstyle{pinstyle} = [pin edge={to-,thin,black}]
    
    \begin{tikzpicture}[auto, node distance = {0.3cm and 0.5cm}]
        \node [input, name=input] {};
        \node [sum, right = of input] (sum) {$\Sigma$};
        \node [block, right = of sum] (lti) {$H_1$};
        \node [coordinate, right = of lti] (z_intersection) {};
        \node [output, right = of z_intersection] (output) {}; 
        \node [block, below = of lti] (static_nl) {$H_2$};
    
        \draw [->] (input) -- node {$r$} (sum);
        \draw [->] (sum) -- node {$e$} (lti);
        \draw [->] (lti) -- node [name=z] {$y$} (output);
        \draw [->] (z) |- (static_nl);
        \draw [->] (static_nl) -| node[pos=0.99] {$-$} (sum);
    \end{tikzpicture}
    
    \caption{Block diagram of a general feedback interconnection.}
    \label{fig:chaffey_thm2}

    \vspace{-0.5em}
\end{figure}

\section{Pitfall of Stability Analysis with SRGs}\label{sec:pitfall}

\subsection{An apparent contradiction}\label{sec:contraction}


We will now argue that Theorem~\ref{thm:lti_srg} is a problematic description of the SRG of an LTI system. Consider the simple feedback setup in Fig.~\ref{fig:linear_feedback}, where $L(s) = \frac{-2}{s^2+s+1}$. Since we work with an LTI system, well-posedness of $T = (L^{-1} + 1)^{-1} : \Lte \to \Lte$ is understood. Chaffey et al. introduced Theorem~\ref{thm:chaffey_thm2} in~\cite{chaffeyGraphicalNonlinearSystem2023}, which includes a homotopy argument, precisely to deal with well-posedness of the system. Therefore, one may expect to analyze the stability of $T$ using only the SRG calculus rules from Section~\ref{sec:srg_definitions} developed by Ryu et al. in~\cite{ryuScaledRelativeGraphs2022}, i.e. without using Theorem~\ref{thm:chaffey_thm2}. However, as argued below, this is not the case.

Firstly, we analyze the system using SRG calculus. Since $L(s) = \frac{-2}{s^2+s+1}$ has poles $s=-1/2 \pm j \sqrt{3}/2$, it is stable, and the $\SRG(L)$ is obtained via Theorem~\ref{thm:lti_srg}, see Fig.~\ref{fig:srg_L1}. The SRG of the closed-loop system, is obtained by first applying Proposition~\ref{prop:srg_calculus}.\ref{eq:srg_calculus_inverse} to obtain $\SRG(L^{-1})$ in Fig.~\ref{fig:srg_L1_inv}. Then, we use Proposition~\ref{prop:srg_calculus}.\ref{eq:srg_calculus_plus_one} to obtain $\SRG(1+L^{-1})$, see Fig.~\ref{fig:srg_L1_inv_p_1}, and finally use Proposition~\ref{prop:srg_calculus}.\ref{eq:srg_calculus_inverse} again to obtain $\SRG(T)$, see Fig.~\ref{fig:srg_T1}. The radius of $\SRG(T)$, as obtained via SRG calculus, is clearly finite. This means that $T$ has finite incremental $L_2$-gain, which would show that $T$ is stable. 

\begin{figure*}[tb]
     \centering
     \begin{subfigure}[b]{0.162\linewidth}
         \centering
         \includegraphics[width=\linewidth]{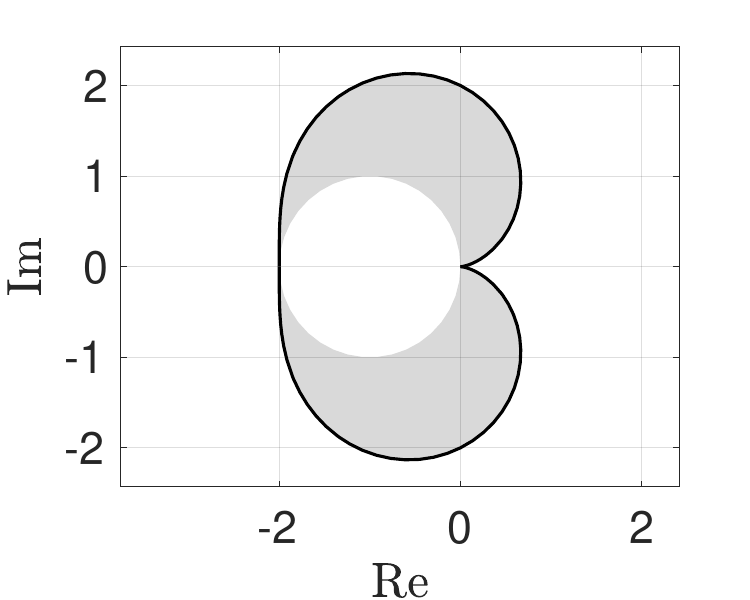}
         \caption{$\SRG(L)$.}
         \label{fig:srg_L1}
     \end{subfigure}
     \hfill
     \begin{subfigure}[b]{0.16\linewidth}
         \centering
         \includegraphics[width=\linewidth]{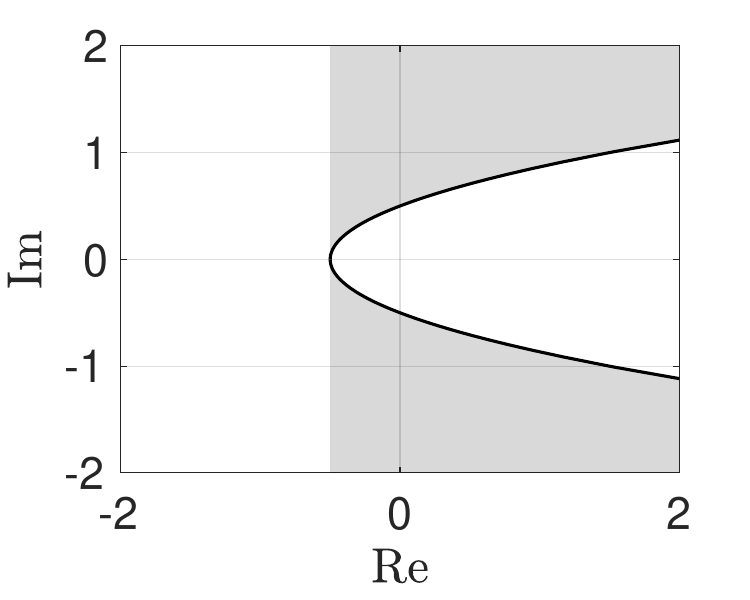}
         \caption{$\SRG(L)^{-1}$.}
         \label{fig:srg_L1_inv}
     \end{subfigure}
     \hfill
     \begin{subfigure}[b]{0.16\linewidth}
         \centering
         \includegraphics[width=\linewidth]{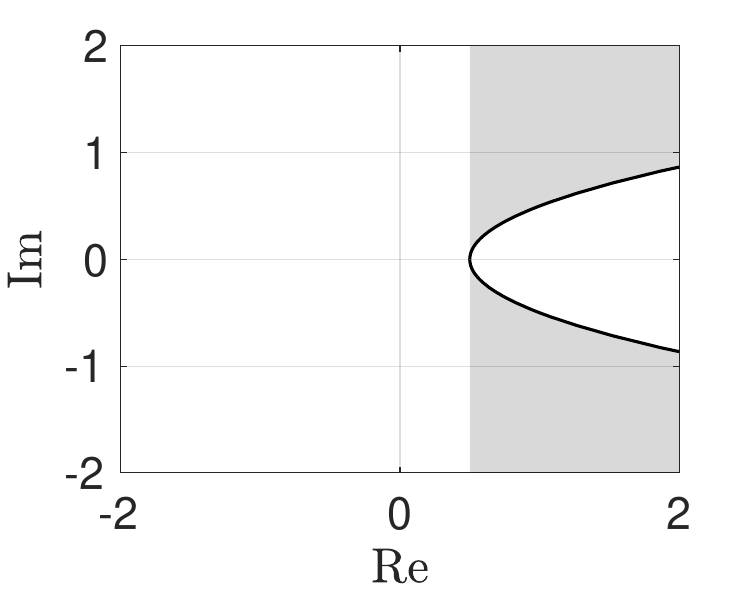}
         \caption{$1 + \SRG(L)^{-1}$.}
         \label{fig:srg_L1_inv_p_1}
     \end{subfigure}
     \hfill
     \begin{subfigure}[b]{0.163\linewidth}
         \centering
         \includegraphics[width=\linewidth]{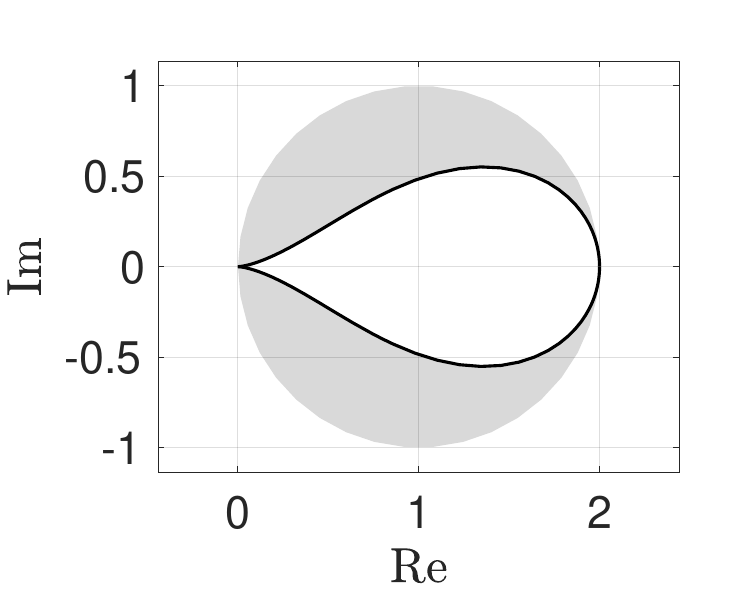}
         \caption{$\SRG_\mathcal{U}(T)$.}
         \label{fig:srg_T1}
     \end{subfigure}
     \hfill
     \begin{subfigure}[b]{0.162\linewidth}
         \centering
         \includegraphics[width=\linewidth]{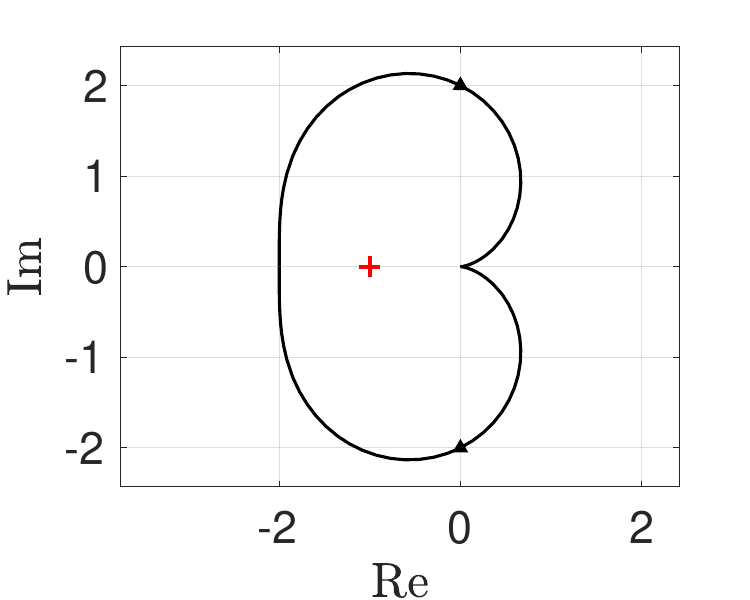}
         \caption{$\operatorname{Nyquist}(L)$.}
         \label{fig:nyquist_L1}
     \end{subfigure}
     \hfill
     \begin{subfigure}[b]{0.162\linewidth}
         \centering
         \includegraphics[width=\linewidth]{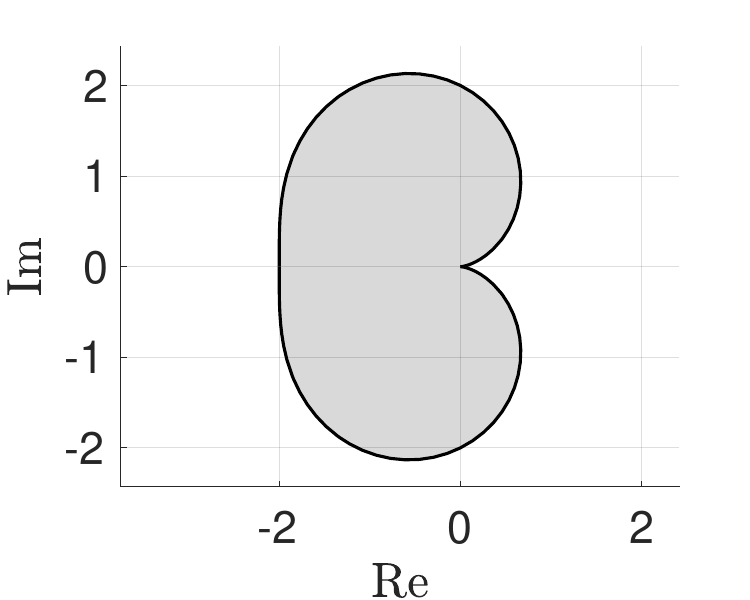}
         \caption{$\SRG'(L)$.}
         \label{fig:srg_L1_extended}
     \end{subfigure}
        \caption{SRGs and Nyquist diagram corresponding to $L(s) = \frac{-2}{s^2+s+1}$. The shaded area is the SRG and the bold line is the Nyquist diagram.}
        \label{fig:srg_analysis_L1_T1}

    \vspace{-1.2em}
\end{figure*}

However, when one uses the Nyquist stability criterion Theorem~\ref{thm:nyquist} to analyze stability, one concludes that $T(s)$ is unstable. This can be seen from the Nyquist diagram in Fig.~\ref{fig:nyquist_L1}, which encircles $-1$ one time in clockwise fashion, which results in $n_\mathrm{z} =1$ in Theorem~\ref{thm:nyquist}, indicating that $T(s)$ has one unstable pole.

It appears that we have derived a \emph{contradiction}. Nyquist theory correctly predicts instability, while using the rules of SRG calculus we arrive at a wrong result. 

This apparent dichotomy is reminiscent of the Nyquist diagram of an unstable plant. When an LTI plant is continuously transformed from stable to unstable, the Nyquist diagram only achieves infinite radius at the transition point between stable and unstable behavior. An example is $T_a(s) = 1/(s+a)$, which is stable for $a>0$, unstable for $a<0$, and achieves infinite radius only at the transition point $a=0$. For LTI systems, the fact that the radius only attains infinity at a transition point is not a problem, since we can use the Nyquist criterion to assess stability. For SRGs, however, we do not have access to this information. Before we explain the solution, we elaborate a bit more on this pitfall, and the SRG of an unstable LTI operator.

\subsection{Understanding the Pitfall}

With the current SRG formulas, the SRG radius blows becomes infinite at the transition point between stable and unstable. Therefore, direct application of SRG formulas, i.e. not using homotopy arguments that start from a stable system, can lead to \enquote{false positives} that incorrectly predict a stable system. This means that the bound for the incremental $L_2$-gain is only valid if we know a priori that the system is stable, or part of an internally stable feedback loop, analogous to the maximum gain that the Nyquist curve predicts. 

When dealing with a Nyquist curve, however, one can count encirclements of $-1$ to check stability. This circle counting is not available for the h-convex hull of a Nyquist curve, since the frequency information is thrown away, and we only have a set of complex numbers.

Now it can be understood why Theorem~\ref{thm:chaffey_thm2} requires stability to begin with, and uses a homotopy argument. It is not only to guarantee well-posedness, but also to prevent \enquote{false positive} predictions.

\subsection{The SRG of an Unstable LTI System}\label{sec:srg_of_unstable_lti}

We know the exact SRG of a stable LTI operator via Theorem~\ref{thm:lti_srg}. One could argue that this theorem can be extended to marginally stable operators, i.e., those containing integrators, using a limiting argument  where poles are represented as $\lim_{a \downarrow 0} 1/(s+a)$. However, it is unclear how to compute the SRG of an unstable LTI operator.

For an unstable LTI plant, the Nyquist diagram, or the Bode diagram for that matter, can only be interpreted as gain and phase per frequency information if the unstable plant is part of an internally stable feedback system. That is, the plant only receives signals that stabilize the plant. Denote $\mathcal{U} \subset \Lte$ the set of signals that stabilize the unstable LTI plant $L$. Then it is clear that $\SRG_\mathcal{U}(L)$ is the h-convex hull of $\operatorname{Nyquist}(L)$. This is precisely what is obtained in Fig.~\ref{fig:srg_T1}, which is $\SRG_\mathcal{U}(T)$, where $\mathcal{U}$ is the set of signals that stabilize $T$. This case is an example of where the SRG constrained to some input space is used. Here, it misleads us since we had that impression that we studied the stability on $\Lte$, whereas we actually only considered $\mathcal{U}$.

\section{Resolution of the Pitfall}\label{sec:resolution}

The fundamental reason of the pitfall reported is that the SRG disregards the information provided by the Nyquist criterion. As a resolution of this pitfall, we prove that the Nyquist criterion can be combined with the SRG, such that direct application of SRG calculus leads to consistent results.

\subsection{Main result}

Before stating our main result, we need the following definition.

\begin{definition}\label{def:extended_srg}
    Let $R$ be an LTI operator with $n_\mathrm{p}$ poles with $\mathrm{Re}(p) >0$. Denote the h-convex hull of $\operatorname{Nyquist}(R)$ as $\mathcal{G}_R$ and define
    \begin{equation}\label{eq:set_of_encircled_unstable_points}
        \mathcal{N}_R = \{ z \in \C \mid N_R(z) +n_\mathrm{p} >0 \},
    \end{equation}
    where the winding number $N_R(z) \in \Z$ denotes the amount of clockwise encirclements of $z$ by $\operatorname{Nyquist}(R)$. 
    Define the extended SRG of an LTI operator as 
    \begin{equation}\label{eq:lti_srg_redefinition}
        \SRG'(R) := \mathcal{G}_R \cup \mathcal{N}_R.
    \end{equation}
\end{definition}

To illustrate Definition~\ref{def:extended_srg}, we have plotted the extended SRG for $L(s) = \frac{-2}{s^2+s+1}$ in Fig.~\ref{fig:srg_L1_extended}. Because the \enquote{hole} in $\SRG(L)$ in Fig.~\ref{fig:srg_L1} is now filled in, one cannot derive the contradiction from Section~\ref{sec:contraction} anymore.

\begin{theorem}\label{thm:lti_srg_nyquist_extension}
    Under the condition that one of $\SRG(G)^{-1}$ or $\SRG(\phi)$ obeys the chord property and there exists some $\kappa \in \R$ such that $\kappa \in \SRG(\phi)$ and for all $0\leq \tau_1 \leq \tau_2 \leq 1$, it holds that 
    \begin{equation}\label{eq:phi_homotopy_assumption}
        \tau_1 \left( \SRG(\phi) - \kappa \right) \subseteq \tau_2 \left( \SRG(\phi) - \kappa \right).
    \end{equation}
    Then, if  
    \begin{equation}\label{eq:separation_assumption}
        \dist(\SRG'(G)^{-1}, -\SRG(\phi)) \geq r >0
    \end{equation}
    holds, the system in Fig.~\ref{fig:lure} is a stable and well-posed operator $T : \Lte \to \Lte$. Furthermore, the closed-loop operator satisfies $\Gamma(T) \leq 1/r$. If $\SG_0(\phi)$ is used instead of $\SRG(\phi)$, the well-posedness property is lost and $\gamma(T) \leq 1/r$ instead.
\end{theorem}

\begin{proof}
    Fix such a real $\kappa \in \SRG(\phi)$ and apply a loop transformation to obtain $\tilde{G} = \frac{G}{1+ \kappa G}$ and $\tilde{\phi} = \phi - \kappa$, such that $T = (G^{-1} + \phi)^{-1} = (\tilde{G}^{-1} + \tilde{\phi})^{-1}$. By Eq.~\eqref{eq:separation_assumption}, $-1/\kappa \notin \SRG'(G)$, hence $\tilde{G}$ is stable according to the Nyquist criterion, therefore, $\Gamma(\tilde{G}) < \infty$.

    Since $\SRG(\tilde{\phi}) = \SRG(\phi) - \kappa$ and $\SRG(\tilde{G}^{-1}) = \SRG(G^{-1}) + \kappa$, Eq.~\eqref{eq:separation_assumption} implies that $\dist(\SRG(\tilde{G})^{-1}, -\SRG(\tilde{\phi})) \geq r$. By Eq.~\eqref{eq:phi_homotopy_assumption}, we know that $\tau \SRG(\tilde{\phi}) \subset \SRG(\tilde{\phi})$ for all $\tau \in [0,1]$, hence $\dist(\SRG(\tilde{G})^{-1}, -\tau \SRG(\tilde{\phi})) \geq r$ for all $\tau \in [0,1]$. 
    
    The rest follows from Theorem~\ref{thm:chaffey_thm2} by taking $H_1 = \tilde{G}$ and $H_2 = \tilde{\phi}$.

\end{proof}

\begin{remark}
    If the condition Eq.~\eqref{eq:phi_homotopy_assumption} does not hold, then Theorem~\ref{thm:lti_srg_nyquist_extension} still holds if $\dist(\SRG'(G)^{-1}, -\tau \SRG(\phi)) \geq r >0$ for all $\tau \in [0,1]$.
\end{remark}

\begin{remark}
    Since only real elements of $\mathcal{N}_G$ in Eq.~\eqref{eq:set_of_encircled_unstable_points} are used in the proof, when invoking the Nyquist criterion, one could consider a smaller extended SRG. Instead of adding $\mathcal{N}_R$ to $\mathcal{G}_R$ in Eq.~\eqref{eq:lti_srg_redefinition}, one could add only $\mathcal{N}_R \cap \R$. This is, however, less appealing and intuitive for graphical analysis of stability via the separation condition Eq.~\eqref{eq:separation_assumption}.
\end{remark}

\begin{remark}
    Theorem~\ref{thm:lti_srg_nyquist_extension} provides a rigorous way of doing SRG computations with plants that have integrators, as opposed to the heuristic method outlined in Section~\ref{sec:srg_of_unstable_lti}.
\end{remark}

\begin{figure*}[t]
    \begin{subfigure}[b]{0.13\linewidth}
        \centering
        \begin{tikzpicture}[scale=0.7]
            \definecolor{myblue}{rgb}{0, 0.1, 0.4}
            \def\axissize{1.5}

            \draw[->] (-\axissize,0) -- (\axissize,0) node[anchor=west] { $x$};
            \draw[->] (0,-\axissize) -- (0,\axissize) node[anchor=east] { $\phi(x)$};

            \draw[-]  (0.5,-0.1) -- (0.5,0.1);

            \draw[-]  (-0.5,-0.1) -- (-0.5,0.1);

            \node at (0.5,-0.25) {\tiny $1$};
            \node at (-0.5,-0.25) {\tiny $-1$};

            \draw[very thick,myblue] (-1, -1.5) -- (-0.5, -0.5) -- (0.5, 0.5) -- (1, 1.5);
        \end{tikzpicture}
        \caption{}
        \label{fig:example_phi}
    \end{subfigure}
    \hfill
    \begin{subfigure}[b]{0.21\linewidth}
        \centering
        \includegraphics[width=\linewidth]{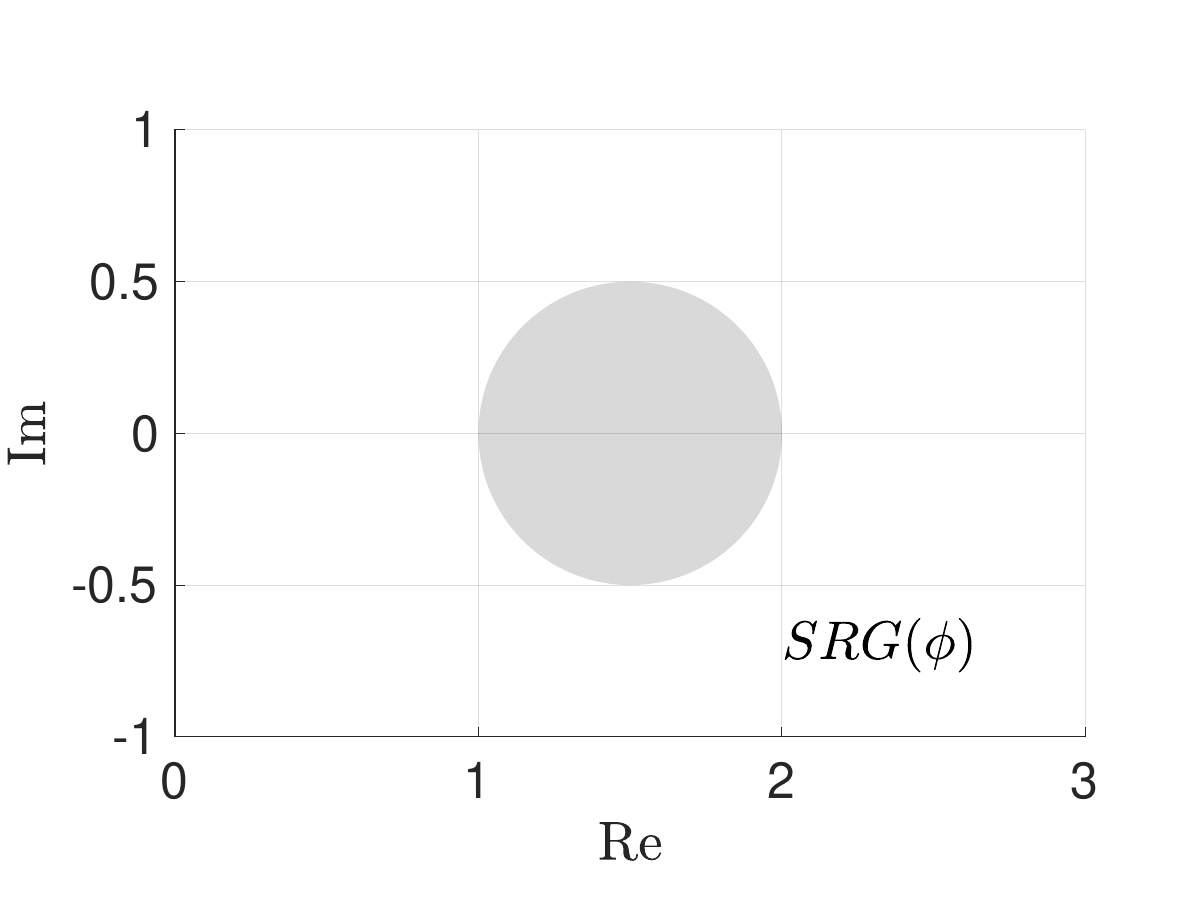}
        \caption{}
        \label{fig:example_srg_phi}
    \end{subfigure}
    \hfill
    \begin{subfigure}[b]{0.2\linewidth}
        \centering
        \includegraphics[width=\linewidth]{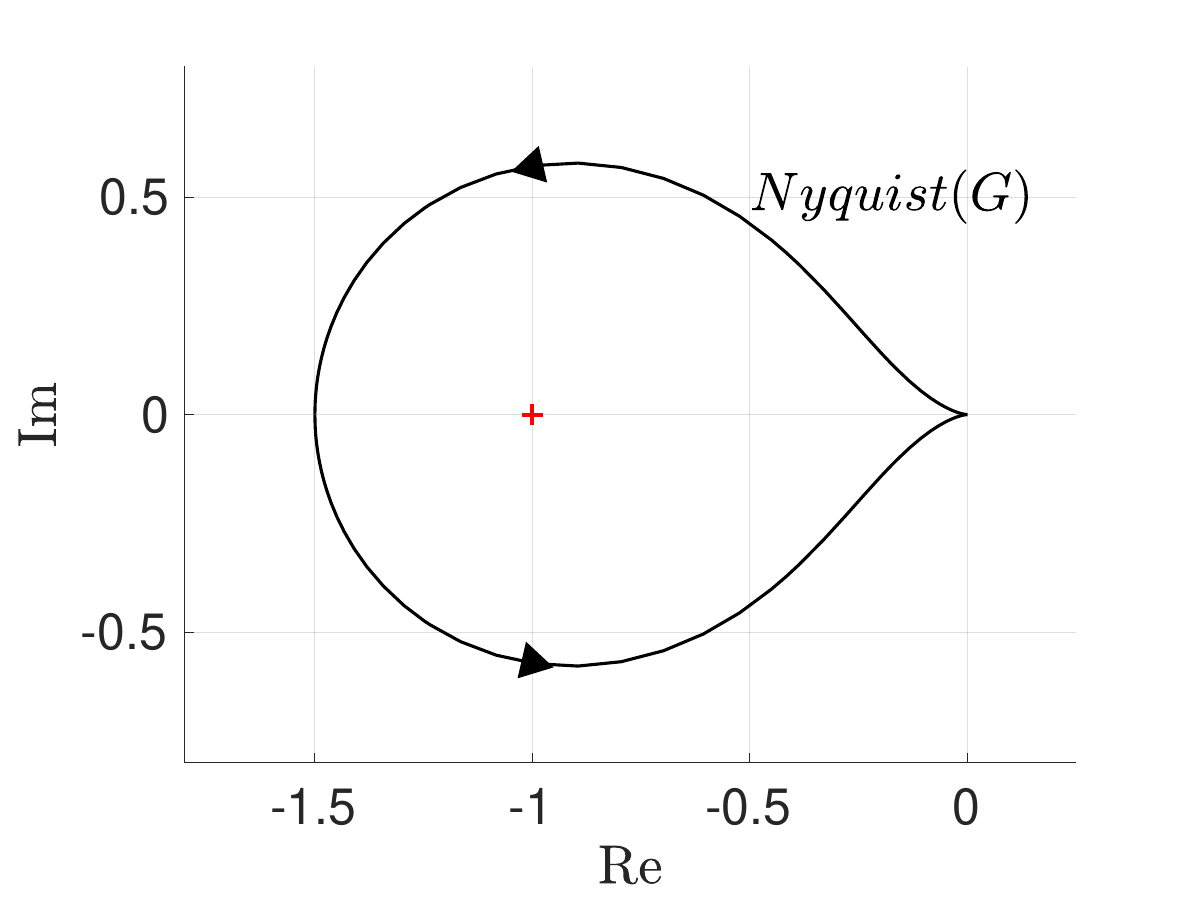}
        \caption{}
        \label{fig:example_nyquist_G}
    \end{subfigure}
    \hfill
    \begin{subfigure}[b]{0.2\linewidth}
        \centering
        \includegraphics[width=\linewidth]{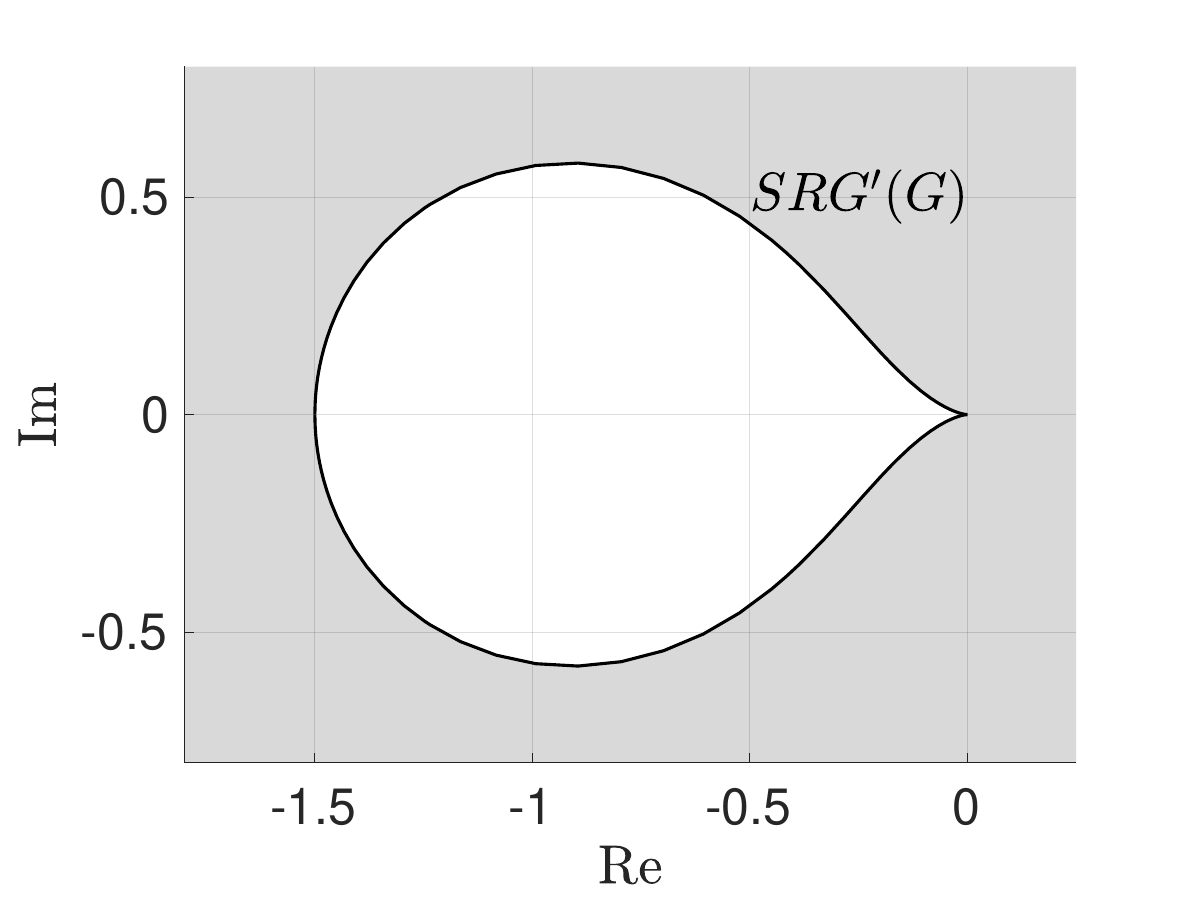}
        \caption{}
        \label{fig:example_srg_G}
    \end{subfigure}
    \hfill
    \begin{subfigure}[b]{0.21\linewidth}
        \centering
        \includegraphics[width=\linewidth]{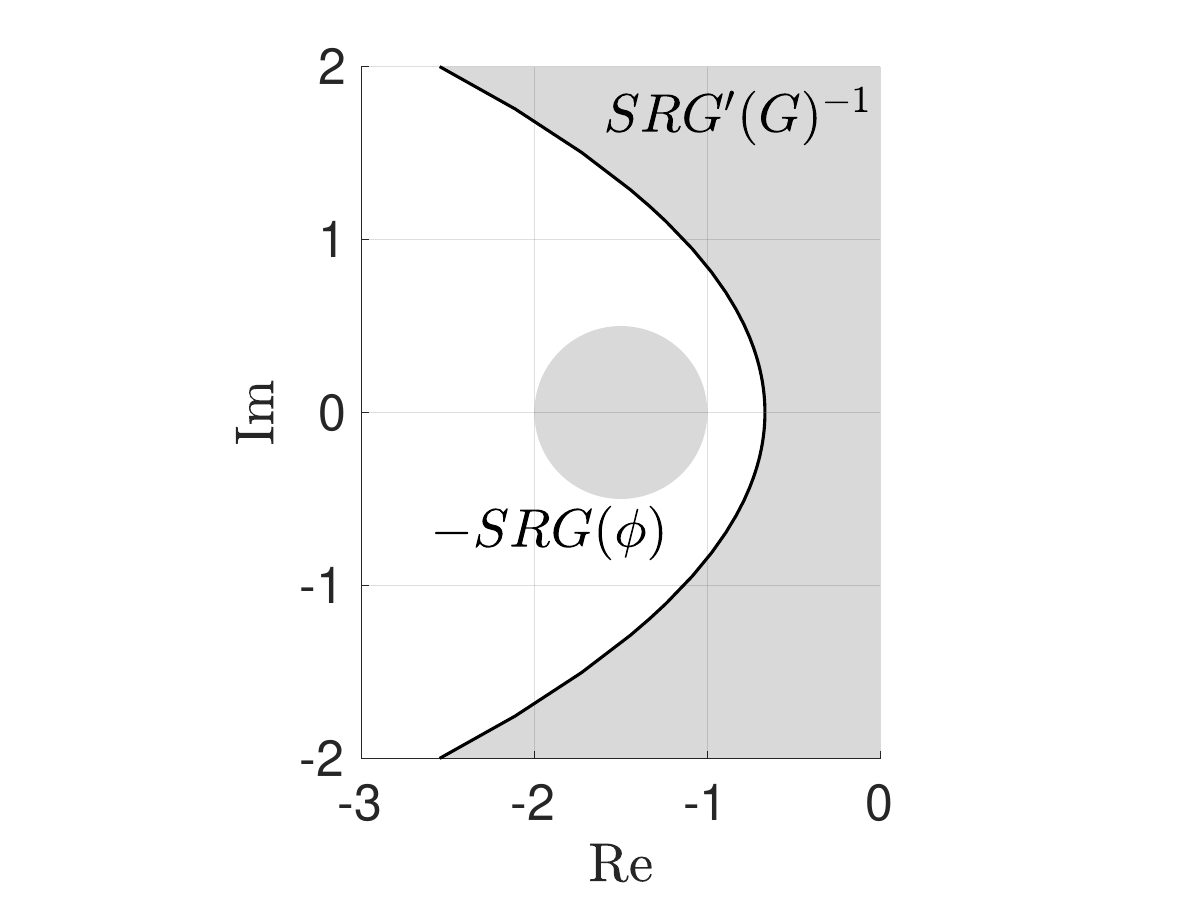}
        \caption{}
        \label{fig:example_srg_distance}
    \end{subfigure}
    \caption{Figures for the SRG analysis of the example in Section~\ref{sec:examples}: (a) graph of the nonlinearity $\phi$, (b) $\SRG(\phi)$, (c) Nyquist diagram of $G$, (d) $\SRG'(G)$, (e) visualization of the separation condition in Eq.~\eqref{eq:separation_assumption}.}
    \label{fig:example}
    \vspace{-1.em}
\end{figure*}

\subsection{The Generalized Circle Criterion}\label{sec:compare_circle_criterion}

Theorem~\ref{thm:chaffey_thm2} can only reproduce the circle criterion in the case that $G(s)$ in Fig.~\ref{fig:lure} is stable. We will now show that using our result, Theorem~\ref{thm:lti_srg_nyquist_extension}, one can prove a result that is more general than the celebrated circle criterion, Theorem~\ref{thm:circle}. We emphasize that this is possible since we combined the information of the Nyquist criterion into the SRG. 

\begin{theorem}\label{thm:generalized_circle_criterion}
    Let $G(s)$ be an LTI operator and $\phi$ a nonlinear operator, where at least one of $\SRG(G)^{-1}$ or $\SG_0(\phi)$ obeys the chord property, and $\SG_0(\phi)$ obeys Eq.~\eqref{eq:phi_homotopy_assumption}. The system in Fig.~\ref{fig:lure} satisfies $r \in \Lte \implies y \in \Lte$ if 
    \begin{equation}\label{eq:circle_criterion_srg}
        \dist(\SRG'(G),-\SG_0(\phi)^{-1}) >0.
    \end{equation}
    Furthermore, we have the $L_2$-gain bound $\gamma(T) \leq 1/r_m$, where $\dist(\SRG'(G)^{-1},-\SG_0(\phi)) \geq r_m>0$.
    
    One can replace the $L_2$-gain with the incremental $L_2$-gain by taking $\SRG(\phi)$ instead of $\SG_0(\phi)$ in~\eqref{eq:circle_criterion_srg}.
\end{theorem}

\begin{proof}
    The theorem is a direct result of Theorem~\ref{thm:lti_srg_nyquist_extension} upon realizing that Eq.~\eqref{eq:circle_criterion_srg} is equivalent to $\dist(\SRG'(G)^{-1},-\SG_0(\phi)) \geq r_m$ for some $r_m>0$.
\end{proof}


Note that Theorem~\ref{thm:generalized_circle_criterion} is equivalent to the circle criterion in Theorem~\ref{thm:circle} when $\phi \in [k_1, k_2]$. However, it is more general than the circle criterion since $\phi$ can be any operator, not necessarily sector bounded. Guiver et al.~\cite{guiverCircleCriterionClass2022} derived a circle criterion for sector bounded dynamic operators, however their method provides no graphical tools to check stability. Additionally, SRG analysis provides a bound on the (incremental) $L_2$-gain of the system, whereas the classical circle criterion and Ref.~\cite{guiverCircleCriterionClass2022} only guarantee boundedness.

\begin{remark}
One might worry that the h-convex hull of the Nyquist diagram will make the SRG analysis more conservative than the circle criterion. We will argue here that this is not the case. Recall that $D_{[k_1, k_2]}$ is a disk centered on the real line, hence so is $D_{[-1/k_1, -1/k_2]} =: D_\phi$. Suppose that the h-convex hull does make the SRG analysis more conservative, i.e.
\begin{equation}\label{eq:nyquist_hconvhull_contradiction}
\begin{aligned}
    &D_\phi \cap \operatorname{Nyquist}(G) = \emptyset, \\
    &D_\phi \cap \SRG(G) \neq \emptyset,
\end{aligned}
\end{equation}
are both true, which means that the circle criterion would predict stability, but the SRG method would be conservative and is not able to predict stability. Note that we use the SRG for $G$ as defined in Theorem~\ref{thm:lti_srg}, so we only consider the h-convex hull part. If~\eqref{eq:nyquist_hconvhull_contradiction} would be true, then there exist $z_1, z_2 \in \operatorname{Nyquist}(G) \cap \C_{\mathrm{Im} >0}$ such that $z_1, z_2 \notin D_\phi$, but $\operatorname{Arc}_\text{min}(z_1, z_2) \cap D_\phi \neq \emptyset$. This is only possible if either $z_1$ or $z_2$ lies in $D_\phi$, since $D_\phi$ is already h-convex. This contradicts~\eqref{eq:nyquist_hconvhull_contradiction}, and therefore the statement is false and we can conclude that the h-convex hull does not make the SRG analysis of the Lur'e system more conservative than the circle criterion.
\end{remark}

\begin{remark}
    Since $G$ is an LTI operator, it is known that $\SRG(G) = \SG_0(G)$, see Ref.~\cite{chaffeyGraphicalNonlinearSystem2023}. Therefore, extending $\SG_0(G)$ with the Nyquist criterion, as in Eq.~\eqref{eq:lti_srg_redefinition} would yield precisely $\SRG'(G)$. For this reason, there is no conservatism introduced by using $\SRG'(G)$ in calculations for non-incremental gain.
\end{remark}

\section{Example: Plant with Unstable Pole}\label{sec:examples}

Consider the Lur'e system in Fig.~\ref{fig:lure} with $G(s) = \frac{3}{(s-2)(s/10+1)}$
\begin{equation*}
    \phi(x) = \begin{cases}
        2x+1 &\text{ if } x<-1, \\
        x &\text{ if } |x| \leq 1, \\
        2x-1 &\text{ if } x>1,
    \end{cases}
\end{equation*}
which is visualized in Fig.~\ref{fig:example_phi}. We are interested in the incremental gain of $T = (G^{-1} + \phi)^{-1}$. For this, we will apply Theorem~\ref{thm:lti_srg_nyquist_extension}. 
While Theorem~\ref{thm:circle} (incremental circle criterion) can be used, it only ensures $L_2$ stability without a gain bound, unlike Theorem~\ref{thm:lti_srg_nyquist_extension}, which provides an incremental $L_2$-gain bound.

Since $\phi$ satisfies the incremental sector condition $\partial \phi \in [1,2]$, we know that $\SRG(\phi) = D_{[1,2]}$, see Fig.~\ref{fig:example_srg_phi}, which obeys the chord property. Moreover, by picking any $\kappa \in [1,2]$, e.g. $\kappa = 1.5 \in \SRG(\phi)$, it is clear that condition~\eqref{eq:phi_homotopy_assumption} holds. To compute $\SRG'(G)$, we start by computing the Nyquist diagram, see Fig.~\ref{fig:example_nyquist_G}. From Fig.~\ref{fig:example_nyquist_G}, it is clear that $N_G(z) + n_\mathrm{p} =0$ in Eq.~\eqref{eq:set_of_encircled_unstable_points} for all points \enquote{inside} the Nyquist diagram, where $n_\mathrm{p}=1$ is the amount of unstable poles. By adding $\mathcal{N}_G$ and $\mathcal{G}_G$, the h-convex hull of the Nyquist diagram, to Fig.~\ref{fig:example_nyquist_G}, one obtains $\SRG'(G)$, shown in Fig.~\ref{fig:example_srg_G}.

To compute the distance in Eq.~\eqref{eq:separation_assumption}, we must use Propostion~\ref{prop:srg_calculus}.\ref{eq:srg_calculus_inverse} to \emph{invert} $\SRG'(G)$. This gives rise to Fig.~\ref{fig:example_srg_distance}, where we can read off that 
\begin{equation*}
    \dist(\SRG'(G)^{-1}, -\SRG(\phi)) = 1/4.
\end{equation*}
Now that we satisfy all the conditions for Theorem~\ref{thm:lti_srg_nyquist_extension}, we can conclude that $T = (G^{-1} + \phi)^{-1}$ is a well-posed operator $T: \Lte \to \Lte$ (provided that $T$ is causal) with incremental gain bound $\Gamma(T) \leq 4$.

\section{Conclusion and Outlook}\label{sec:conclusion}

We have shown how to combine the Nyquist criterion with the SRG for LTI operators in order to use SRG analysis with unstable LTI plants. One immediate result is the generalized circle criterion, and we have demonstrated its application to an example. We focused on the Lur'e setup beyond sector bounded nonlinearities, for which we also guarantee well-posedness in the incremental case. Topic of further research is to generalize Theorem~\ref{thm:lti_srg_nyquist_extension} to more general feedback interconnections, while preserving well-posedness, to fully overcome the limitation of current SRG methods.

\bibliographystyle{IEEEtran} 
\bibliography{bibliography} 

\end{document}